\documentclass[a4paper]{article}
\usepackage[a4paper, total={6.5in, 8in}]{geometry}
\usepackage{bm}
\usepackage{float}
\usepackage[T1]{fontenc}
\usepackage{changepage}
\usepackage{algpseudocode}
\usepackage{booktabs}
\usepackage[ruled]{algorithm}
\usepackage{graphicx}
\usepackage[shortlabels]{enumitem}

\def\Return{\State\textbf{return}\ }

\newcommand{\radeq}{\smash{\stackrel{\mathrm{\footnotesize rad}}{=}}}

\newcommand{\ol}[1]{\overline{#1}}
\newcommand{\lt}{\operatorname{lt}}
\newcommand{\lm}{\operatorname{lm}}
\newcommand{\lcm}{\operatorname{lcm}}
\newcommand{\sig}{\mathfrak{s}}

\newcommand{\ind}{\operatorname{ind}}

\newcommand{\sat}[2]{(#1:#2^{\infty})}

\usepackage{amsmath,amsfonts,amssymb,amsthm}
\newtheorem{definition}{Definition}
\numberwithin{definition}{section}
\newtheorem{theorem}[definition]{Theorem}

\newtheorem{proposition}[definition]{Proposition}
\newtheorem{lemma}[definition]{Lemma}
\newtheorem{example}[definition]{Example}
\theoremstyle{remark}
\newtheorem{remark}{Remark}[section]

\usepackage[
citestyle=numeric,
bibstyle=numeric,
backend=biber,
isbn=false,
maxcitenames=2,
maxbibnames=6,
giveninits=true,
sortcites=true,
url=false,
doi=false,
]{biblatex}
\bibliography{sgb_nondegen.bib}

\AtEveryBibitem{\clearfield{month}}
\AtEveryBibitem{\clearfield{day}}
\DeclareNameAlias{sortname}{first-last}

\def\st{\ \middle|\ }
\def\poly{\operatorname{poly}}
\def\quo{\operatorname{quo}}

\title{A Signature-based Algorithm for\\ Computing
  the Nondegenerate Locus of a Polynomial System\footnote{This work has been supported by European Union's Horizon 2020 research and innovation programme under the Marie Skłodowska-Curie Actions, grant agreement 813211 (POEMA)
  by European Research Council under the European Union's Horizon Europe research and innovation programme, grant agreement 101040794 (10000~DIGITS);
  by the joint ANR-FWF grant ANR-19-CE48-0015 (ECARP),
  the ANR grant ANR-19-CE40-0018 (De Rerum Natura), the DFG
  Sonderforschungsbereich TRR 195, and the Forschungsinitiative Rheinland-Pfalz.
}}

\author{Christian Eder \thanks{TU Kaiserslautern, Germany}
  \and Pierre Lairez \thanks{Inria, Uni. Paris-Saclay, Palaiseau, France}
  \and Rafael Mohr \footnotemark[2] \footnotemark[4]
  \and Mohab Safey El Din \thanks{Sorbonne Uni., CNRS, Paris, France}
  }

\begin{document}
\maketitle 

\begin{abstract}
  Polynomial system solving arises in many application areas to model non-linear
  geometric properties. In such settings, polynomial systems may come with
  degeneration which the end-user wants to exclude from the solution set. The
  nondegenerate locus of a polynomial system is the set of points where the
  codimension of the solution set matches the number of equations.

  Computing the nondegenerate locus is classically done through ideal-theoretic
  operations in commutative algebra such as saturation ideals or equidimensional
  decompositions to extract the component of maximal codimension.

  By exploiting the algebraic features of signature-based Gr\"obner basis
  algorithms we design an algorithm which computes a Gröbner basis of the
  equations describing the closure of the nondegenerate locus of a polynomial system, without
  computing first a Gröbner basis for the whole polynomial system.
\end{abstract}
\section{Introduction}

\paragraph{Problem Statement}
Fix a field $\mathbb{K}$ with an algebraic closure $\ol{\mathbb{K}}$ and a polynomial ring
$R:=\mathbb{K}[x_1,\dots,x_n]$ over $\mathbb{K}$. Let $f_1,\dots,f_c\in R$ and
$V:=\big\{ p\in \ol{\mathbb{K}}^n \ \big|\ f_1(p)=\dots=f_c(p)=0 \big\}$.
Further define the ideal $I:=\langle f_1,\dots, f_c\rangle = \{\sum_{i=1}^c q_i
f_i \mid q_i \in R\}$.
The algebraic set $V$ is a finite union of irreducible components.
By the
Principal Ideal Theorem \cite[Theorem~10.2]{eisenbud_comm-alg} the
codimension of the $\mathbb{K}$-irreducible components of $V$ is at most $c$.
Let $V_c$ denote the union of the components of~$V$ of codimension exactly~$c$.
In particular~$V_c = \varnothing$ when~$c > n$.

The goal of this paper is to compute a Gr\"obner basis of an ideal whose zero
set is $V_c$, which we call the \emph{nondegenerate locus} of the system $f_1,
\ldots, f_c$ (note that we may not compute a \emph{radical} ideal).

{\em Prior works and scientific locks.} 
State-of-the-art algorithms to compute the nondegenerate locus of $f_1,\dotsc,f_c$ rely on the more general problem of
computing the equidimensional decomposition of the ideal that they generate.
There is a vast body of literature split along what data structure is used for the output into two research lines.

The first family of algorithms computes a Gröbner basis of the ideal of each component.
There are two different approaches in this line.
The first uses projections, computed with \emph{elimination orderings}, to reduce the problem to a problem for hypersurfaces \cite{GianniTragerZacharias_1988,KrickLogar_1991,CaboaraContiTraverse_1997}.
The second relies on homological characterizations of the dimension and the computation of free resolutions \cite{EisenbudHunekeVasconcelos_1992}.
See \cite{DeckerGreuelPfister_1999,GreuelPfister_2007,Vasconcelos_1998} and
references therein for further references.
Both approaches use Gröbner basis algorithms as a black box for performing various ideal-theoretic operations, in particular ideal quotients (also known as colon ideals).

A second family of algorithms outputs equidimensional components of $I$ or its
radical through \emph{lazy} representations, i.e. as complete intersections over a
non-empty Zariski open set. This is the case for the so-called regular chains
which go back to Wu-Ritt characteristic sets \cite{Wu_1986}.

These put into practice a kind of D5 principle \cite{DellaDoraDicrescenzoDuval_1985} to split geometric objects by
enforcing an equiprojectability property.
See~\cite{Hubert_2003,Wang_1993,Wang_2001,ChouGao_1990,AubryLazardMorenoMaza_1999,LemaireMorenoMazaPanXie_2011}
and references therein for further references. When the base field $\mathbb{K}$ has characteristic $0$ (or large enough characteristic), geometric resolution
algorithms \cite{GiustiLecerfSalvy_2001} can also be used. These culminate with
the incremental algorithm in \cite{lecerf_geom_resolution,Lecerf_2003} which
avoids equiprojectability issues by performing a linear change of variables to ensure
Noether position properties. One feature is that input polynomials are encoded
with straight-line programs to take advantage of evaluation properties. See also \cite{JeronimoSabia_2002} for a similar approach. It also gives the best known
complexity for equidimensional decomposition: linear in the evaluation
complexity of the input system and polynomial in some algebraic degree.

As of software, the computer algebra systems Singular \cite{Singular_CAS}, Macaulay2~\cite{M2} and Magma~\cite{BosmaCannonPlayoust_1997}  implement
the algorithm of \cite{EisenbudHunekeVasconcelos_1992}
to perform equidimensional decomposition. Maple
implements algorithms for computing regular chains \cite{DMSWX05, DJMS08, CGLMP07, CLMPX07, LMP09} and algorithms
based on Gr\"obner bases. The algorithm by Gianni et
al.~\cite{GianniTragerZacharias_1988} is used for prime
decomposition and, combined with techniques from~\cite{BWK93}, for
equidimensional decomposition. All these implementations use Gr\"obner basis
algorithms as a black box.

\paragraph{Main results}
By contrast with previous work, we only focus on computing the nondegenerate locus of a system, not the full equidimensional decomposition of the corresponding ideal. The main difference to other Gröbner basis based techniques to compute equidimensional decompositions is that we enlarge $I$ \emph{while} a Gröbner basis for $I$ is computed and return a Gröbner basis of a nondegenerate locus of $I$ when this Gröbner basis computation is finished.
Modifying or splitting the ideal in question in the middle of Gr\"obner
basis algorithms is a natural and appealing idea \cite[e.g.][]{Grabe97}.

This idea requires one to answer {\em (i)} when the ideal in question should be enlarged
and {\em (ii)} how to
minimize the cost of enlarging the ideal in question. The algorithm we propose
tackles both issues.

We tackle problem {\em(i)} by following the incremental structure of the sGB algorithms on which our work is based
\cite{faugere_f5, g2v-algorithm}. We describe this sGB algorithm in section \ref{sec:f5-algorithm}. Incremental means here that these algorithms proceed by computing first a Gröbner basis
for $\langle f_1,f_2\rangle$ then use the result to compute a Gröbner basis for $\langle f_1,f_2,f_3\rangle$ and so on. In addition, sGB algorithms
keep track of an auxiliary data structure, called a \textit{signature}, which is attached to each considered
polynomial. This enables one to \emph{exclude} certain polynomials from the set of polynomials to be processed by reduction
in Buchberger's algorithm.

As a consequence they have the feature that, having computed a Gröbner basis for $I_{i-1}:=\langle f_1,\dots,f_{i-1}\rangle$, a reduction to
zero happens in the Gröbner basis computation for $I_i$ if any only if
$f_i$ is a zero divisor modulo $I_{i-1}$. In this case $V_{i-1}:=V(I_{i-1})$ has irreducible components on which $f_i$ is identically
zero (the union of which is henceforth denoted $V_{i-1,f_i=0}$)
and components which are not contained in the hypersurface $V(f_i)$ (the union of which is henceforth denoted $V_{i-1,f_i\neq 0}$).
Assuming that $V(I_{i-1})$ is equidimensional of codimension $i-1$, to compute an ideal representing the
nondegenerate locus of $V(I_i)$ we may then proceed as follows (see Algorithm \ref{alg:decomp}):
\begin{enumerate}
\item Compute ideals representing $V_{i-1,f_i=0}$ and $V_{i-1,f_{i}\neq 0}$ (via the ideal-theoretic operation of \textit{saturation}).
\item Compute an ideal representing $W:=V_{i-1,f_{i}\neq 0}\cap V(f_i)$.
\item Remove from $W$ all components contained in $V_{i-1,f_i=0}$ (again via saturation).
\end{enumerate}

Iterating over the set of input equations with these three steps, using the result of each
iterative step as input for the next invocation of this loop and slightly adapting the third step to remove all components which are contained in components of higher dimension then yields an ideal representing
the nondegenerate locus of $I$. We describe this algorithm from a purely algebraic perspective in section \ref{sec:core_loop}.

To tackle problem \emph{(ii)} we exploit a feature of the incremental sGB algorithms first captured in the G2V algorithm
\cite{g2v-algorithm}: The data of a signature can be enlarged so as to \emph{simultaneously} compute
a Gröbner basis for $I_i$ and the quotient ideal $(I_{i-1}:f_i):=\{g\in R\;|\;gf_{i}\in I_{i-1}\}$ (which, if $I_{i-1}$ is a radical ideal, corresponds precisely to
$V(I_{i-1,f_i\neq 0})$) in each incremental step.

Using this idea we modify the baseline sGB algorithm we use to simultaneously perform steps 1 and 2 of the above loop (i.e. in
a single Gröbner basis computation).
This is done essentially by immediately inserting an element $g\in (I_{i-1}:f_i)$ once it is identified during the run of the sGB algorithm. We manage this insertion of elements that do not lie in the original ideal $I$ with a data structure we call an \emph{sGB tree} (see section \ref{sec:f5-tree-datastr}) which allows us to
perform this modification with the needed technical properties of signatures ensured.
This yields a signature-based version of Algorithm \ref{alg:decomp}, Algorithm \ref{alg:nondeg-f5tree}. Besides managing the insertion of new generators into some initial ideal, the sGB data structure also leaves open the future possibility of designing signature-based ideal decomposition algorithms.

We finally show experimentally in section \ref{sec:experiments} that the consequence of this simple modification is a massive cost reduction in the overhead
compared to a ``naive'' implementation of Algorithm \ref{alg:decomp} where one uses saturation procedures as a blackbox. As is also shown, it additionally enables us to compute the nondegenerate locus of systems which are out of the reach of equidimensional decomposition algorithms available in state of the art computer algebra systems.

\section{The basic algorithm}
\label{sec:core_loop}

Consider a codimension~$k$ irreducible variety~$X\subseteq \ol{\mathbb{K}}^n$
and a polynomial~$f \in R$.
Either~$X \subseteq V(f)$, and so~$X\cap V(f) = X$, or~$X\cap V(f)$ is equidimensional of codimension~$k+1$ (that is, all the irreducible components of~$X\cap V(f)$ have codimension~$k+1$).
If~$X$ is not irreducible, then $X\cap V(f)$ may not be equidimensional. Yet, the alternative above applies to each irreducible component of~$X$.
The components of~$X$ which are not included in~$V(f)$ are exactly the components of the closure of~$X \setminus V(f)$,
while the components of~$X$ which included in~$V(f)$ are exacly the components of the closure of~$X\setminus\ol{(X\setminus V(f))}$.
This leads to the decomposition of~$X\cap V(f)$ as the union of two equidimensional varieties of codimension~$k$ and~$k+1$ respectively:
\[ X = \ol{\left( X \setminus \ol{X \setminus V(f)} \right)} \cup \left( \ol{X\setminus V(f)} \cap V(f) \right). \]
This is the basic identity that we leverage to compute, incrementally, the codimension~$c$ components of an ideal~$\langle f_1,\dotsc,f_c\rangle$.
In an ideal theoretic language, this reformulates as follows.

For two ideals~$I, J \subseteq R$, we write~$I \radeq J$ for the equality of the radicals $\sqrt{I} = \sqrt{J}$.
An ideal~$I$ is \emph{equidimensional} if all the irreducible components of~$V(I)$ have the same dimension.
Recall that~$I : J$ is the ideal~$\left\{ p\in R \st pJ \subseteq I \right\}$.
Recall also that~$I : J^k$ yields an increasing sequence of ideals as~$k\to \infty$, so it eventually stabilizes in an ideal denoted~$I : J^\infty$, the \emph{saturation of~$I$ by~$J$}.
If~$J$ is generated by a single element~$f$, it is simply denoted~$I:f^\infty$.

\begin{lemma}
  \label{lem:decomp}
  For any ideal $J\subseteq R$ and any $f\in R$ we have
  \begin{align*}
    J + \langle f \rangle \radeq \big(\sat{J}{f} + \langle f \rangle\big) \cap \big(J : \sat{J}{f}\big)
  \end{align*}
  Moreover,
  if $J$ is equidimensional of codimension $c < n$ and if $f\notin \sqrt{J}$ then $\left(\sat{J}{f} + \langle f \rangle\right)$ is equidimensional
  of codimension $c+1$ and $(J : \sat{J}{f})$ is equidimensional of codimension $c$.
\end{lemma}

\begin{proof}
  For the left-to-right inclusion, it is clear that~$J$ is included in the right-hand side, so it remains to check that~$f$ is in the radical of both terms of the intersection. It is obvious that~$f \in (J:f^\infty) + \langle f \rangle$, so it remains to prove that~$f$ is in the radical of~$J : (J: f^\infty)$. So let~$g \in J : f^\infty$, that is~$g f^\ell \in J$ for some~$\ell \geq 0$,
  which we may rewrite as~$f \in \sqrt{J : g}$.
  To conclude, we observe that
  \[ \sqrt{J : (J : f^\infty)} = \bigcap_{g\in J: f^\infty} \sqrt{J : g},\]
  so~$f\in  \sqrt{J : (J : f^\infty)}$.

  Conversely,  let $p \in \left(\sat{J}{f} + \langle f \rangle\right) \cap (J : \sat{J}{f})$. Write
  $p = q + af$
  where $q\in \sat{J}{f}$ and $a\in R$.
  Since~$p \in ( J : \sat{J}{f})$, we have $pq \in J$.
  But~$p q = q^2 + aqf$, so~$q^2 \in J + \langle f \rangle$. It follows that~$q \in \sqrt{J+\langle f \rangle}$,
  thus proving the stated equality.

  For the statement on equidimensionality,
  we may rely on the geometric interpretation above:
  the zero set of~$J : f^\infty + \langle f \rangle$ is $\ol{X\setminus V(f)} \cap V(f)$, where~$X = V(J)$.
  The equidimensionality of~$J:(J:f^\infty)$ is slightly more technical because the geometric interpretation only gives information on~$J : (J:f^\infty)^\infty$.
  Yet, both are equal up to radical:
  $V \left( J:(J:f^\infty) \right)$ is the union of the component of~$V(J)$ that are included in~$V(f)$ \cite[Proposition~23]{IshiharaYokoyama_2018}.
\end{proof}

We are now ready to describe Algorithm~\ref{alg:decomp}. To do this we
suppose for now that we have an algorithm for computing the quotient ideal~$J:K$ and the saturation~$J : K^\infty$,
given generators for~$J$ and~$K$.
Given $c \leq n$ elements $f_1,\dots,f_c\in R$
the core loop of Algorithm~\ref{alg:decomp} starts with the ideal $J =
\langle f_1 \rangle$ and to continously replace it with $\sat{J}{f_k}+f_k$ for
each $k$. By
Lemma~\ref{lem:decomp} the resulting ideal will be equidimensional of
codimension $c$. Note however that it may have components that
the original ideal $I=\langle f_1,\dots,f_c\rangle$ does not have, as shown by the following example.
In the algorithm, these additional components are removed with saturations at every iterative step with the loop on line~\ref{line:decom:cleaning}.

\begin{example}
  Let $R = k[x,y,z]$ and $f_1= xy$, $f_2 = xz$. Then
  $\sat{xy}{xz} + xz = \langle y, xz \rangle$ which has the component $\langle y, x \rangle$
  which is not a component of $\langle f_1,f_2 \rangle = \langle x \rangle \cap \langle y, z \rangle$.
\end{example}

\begin{algorithm}[tp]
  \caption{Computation of the nondegenerate locus}
  \label{alg:decomp}
  \begin{algorithmic}[1]
    \Require A set of generators $f_1,\dots,f_c$ for an ideal $I$ in $R$ where $c\leq n$
    \Ensure A set of generators $G$ for the nondegenerate part of $f_1,\dots,f_c$
    \State $J \leftarrow 0$, as an ideal of~$R$
    \State $\mathcal{K} \leftarrow \varnothing$
    \For{$k \in \{1,\dots,c\}$}
    \State $H \leftarrow J : f_k^\infty$
    \State $\mathcal{K} \leftarrow \mathcal{K}\cup \{J : H\}$
    \State $J \leftarrow H + \langle f_k\rangle$
    \For{$K\in \mathcal{K}$} \label{line:decom:cleaning}
    \State $J \leftarrow J : K^\infty$
    \EndFor
    \EndFor
    \Return $J$
  \end{algorithmic}
\end{algorithm}

To prove the correctness of Algorithm \ref{alg:decomp} we also need the following proposition:

\begin{lemma}
  \label{prop:another_eqn}
  For any ideals $I, J\subseteq R$ and any~$f\in R$, we have
  \begin{enumerate}[label=(\roman*)]
    \item\label{it:inter} $(I\cap J) + \langle f \rangle \radeq \left( I + \langle f \rangle \right) \cap \left( J + \langle f \rangle \right)$;
    \item\label{it:intersat} $I \cap J \radeq \left( I : J^\infty \right) \cap J$;
    \item\label{it:satwitheq} if~$f\in I$, then $I : J^\infty = I : (J + \langle f \rangle)^\infty$.
  \end{enumerate}
\end{lemma}

\begin{proof}
  For the first item,
  \begin{align*}
    (I + \langle f \rangle)\cap (J + \langle f \rangle) \radeq (I + \langle f \rangle)(J + \langle f \rangle)
                                                        \radeq IJ + \langle f \rangle \radeq (I\cap J) + \langle f \rangle.
  \end{align*}
  For the second one, the left-to-right inclusion is clear. Conversely, let~$f \in (I : J^\infty) \cap J$
  and let~$k > 0$ such that~$f J^k \in I$.
  In particular~$f^{k+1} \in I$. So~$f\in \sqrt{I}$.

  For the last item is trivial from the definition of saturation.
\end{proof}

\begin{theorem}
  \label{thm:correctness}
  On input~$f_1,\dotsc,f_c \in R$ with $c \leq n$, Algorithm~\ref{alg:decomp} terminates
  and outputs an ideal~$J$ such that~$V(J)$ is the nondegenerate locus of the input system.
\end{theorem}

\begin{proof}
  We define~$J_0 := \langle 0 \rangle$, and then, by induction on~$i$,
  \begin{align*}
    K_i&:=(J_{i-1}:\sat{J_{i-1}}{f_i}),\\
    \text{ and } J_i &:=\bigg(\sat{J_{i-1}}{f_i} + \langle f_i \rangle : \big( {\textstyle\prod_{j=1}^{i}K_j} \big)^\infty \bigg).
  \end{align*}
  It is clear that Algorithm~\ref{alg:decomp} returns the ideal~$J_c$.
  Now, let~$I_i = \langle f_1,\dotsc,f_i \rangle$.
  The main loop invariant, that we prove by induction on~$i$,
  is
  \begin{equation}\label{eq:4}
    I_i \radeq J_i \cap \bigcap_{j=1}^{i}\left(K_j + \langle f_{j+1},\dots,f_{i}\rangle\right) .
  \end{equation}
  From this, we deduce that the zero set of $J_c$ is contained in the
  algebraic set defined by $f_1=\cdots=f_c=0$. We will prove later that
  $J_c$ is equidimensional of codimension $c$, that the components of the ideals $\left(K_j +
    \langle f_{j+1},\dots,f_{c}\rangle\right)$ have codimension less than~$c$
  and do not contain any components of~$J_c$.

  It is trivially true that \eqref{eq:4} holds
  for~$i=0$. For~$i > 0$, we have
  \begin{align*}
    I_i &= I_{i-1} + \langle f_i \rangle = J_{i-1} \cap  \bigcap_{j=1}^{i-1}\left(K_j + \langle f_{j+1},\dots,f_{i-1}\rangle\right)   + \langle f_i \rangle\\
        &\radeq  \left( J_{i-1} + \langle f_i \rangle \right) \cap \bigcap_{j=1}^{i-1}\left(K_j + \langle f_{j+1},\dots,f_{i}\rangle\right), \quad\text{by Lemma~\ref{prop:another_eqn}\ref{it:inter}.}
  \end{align*}
  Besides, by Lemma~\ref{lem:decomp},
  \begin{align*}
    J_{i-1} + \langle f_i \rangle &\radeq \left( \left( J_{i-1} : f_i^\infty \right) + \langle f_i \rangle \right) \cap \left( J_{i-1} : (J_{i-1} : f_i^\infty) \right)\\
    &=  \left( \left( J_{i-1} : f_i^\infty \right) + \langle f_i \rangle \right) \cap K_i.
  \end{align*}
  For short, let~$J'_i = \left( J_{i-1} : f_i^\infty \right) + \langle f_i
  \rangle$. Combining the equalities above, we have
  \begin{align*}
    I_i &\radeq J'_i \cap \bigcap_{j=1}^i \left(K_j + \langle f_{j+1},\dots,f_{i}\rangle\right)\\
        &\radeq  \bigg( J'_i : \big( {\textstyle \prod_{j=1}^i K_i} \big)^\infty \bigg) \cap   \bigcap_{j=1}^i \left(K_j + \langle f_{j+1},\dots,f_{i}\rangle\right),
  \end{align*}
  using Lemma~\ref{prop:another_eqn}\ref{it:intersat} and~\ref{it:satwitheq} (note that~$f_1,\dotsc,f_i \in H_i$).
  This last equality is exactly~\eqref{eq:4}.

  Now, we analyze the dimensions and show that~$V(J_i)$
  is exactly the nondegenerate locus of~$f_1,\dotsc,f_i$.
  Indeed, using Lemma~\ref{lem:decomp}, we check by induction on~$i$ that~$J_i$ is equidimensional of codimension~$i$ (unless~$J_i = \langle 1\rangle$) and that~$K_i$ is equidimensional of codimension~$i-1$ (unless~$K_i = \langle 1\rangle$).
  It follows that all the components of~$K_j + \langle f_{j+1},\dotsc,f_i\rangle$ have codimension at most~$i-1$.
  Moreover, no component of~$J_i$ is included in any~$K_j$, for~$j \leq i$, since~$J_i$ is saturated by the~$K_j$.
  Therefore, using~\eqref{eq:4}, the codimension~$i$ components of~$I_i$ are
  exactly the components of~$J_i$.

  Hence, we deduce that $J_c$ is equidimensional of codimension $c$ whose
  components are not contained in the ones of
  $K_j + \langle f_{j+1},\dotsc,f_c\rangle$, the components of which have
  codimention less than~$c$. Besides, we already observed that its zero set is
  contained in the one defined by the input polynomials $f_1, \ldots, f_c$.
  Since~\eqref{eq:4} holds, we conclude that $V(J_c)$ is the nondegenerate locus
  of the input system.
\end{proof}

\section{Signature-based Gröbner basis computations}
\label{sec:f5-algorithm}

We will rely on the theory of signature-based Gröbner bases in order to implement efficiently Algorithm~\ref{alg:decomp}.

\subsection{Signatures and extended sig-poly pairs}
\label{sec:sign-extend-sig}


We fix in the following a monomial order on $R$ and a sequence of polynomials
$f_1,\dots,f_r \in R$.
Let $I:=\langle f_1,\dots,f_r\rangle$ and $I_i:=\langle f_1,\dots,f_{i}\rangle$.
We describe an algorithm which computes simultenously a Gröbner basis for $I$ and presents the following features:
\begin{enumerate}
\item It computes a Gröbner basis for $I$ incrementally, i.e. first for $\langle f_1\rangle$ then for $\langle f_1,f_2\rangle$ etc.
\item It simultaneously computes Gröbner bases for each ideal $(\langle f_1,\dots,f_{i-1}\rangle : f_i)$, $i=2,\dots,r$.
\end{enumerate}
This algorithm belongs to the class of so called \textit{signature based} Gröbner basis algorithms, the first of which was the F5 algorithm presented in \cite{faugere_f5}. Since then the class of signature-based algorithms has been greatly extended, see \cite{eder_faugere_survey} for a survey.
The idea of leveraging signature-based algorithms to compute simultaneously some colon ideals first appeared in \cite{g2v-algorithm}.
The algorithm we present here is closely related, with some elements from the F5 algorithm. The algorithm presented in this section is fully encompassed by the general algorithmic framework presented in \cite{eder_faugere_survey}.

We start by defining signatures.

\begin{definition}
A \emph{signature} is a pair $\sigma = (i, m)$
of an index in~$\{1,\dots,r\}$ and a monomial in~$R$.
The first component is called the \emph{index}, and denoted~$\ind(\sigma)$.
The second component is called the \emph{monomial part} of~$\sigma$.
\end{definition}

We order the signatures lexicographically, i.e. by writing
\begin{align*}
  (i, m) < (j, n) \Leftrightarrow i < j\text{ or }i=j\text{ and }m < n.
\end{align*}

The product of a monomial~$a \in R$ and a signature~$\sigma = (i, h)$
is defined by~$a \sigma = (i, ah)$.
A signature~$\sigma$ divides another signature~$\tau$ if there is a monomial~$a$
such that~$a \sigma = \tau$, so in particular $\ind(\sigma) = \ind(\tau)$.

The possible indices of a signature are the indices of the input equations.
This relation between the index of a signature and one of the equations $f_i$ is made stronger by the following object:

\begin{definition}\label{def:extended-sig-poly}
  An \emph{extended sig-poly pair}
  is a triple~$\alpha = (f, \sigma, h)$,
  where~$f, h \in R$ and $\sigma$ is a signature
  such that~$\lm(h)$ is equal to the monomial part of~$\sigma$.
  The first component~$f$ is called the \emph{polynomial part} of~$\alpha$, denoted~$\poly(\alpha)$,
  the second component~$\sigma$ is called the \emph{signature}, denoted~$\sig(\alpha)$,
  and the third component is called the \emph{quotient}, denoted~$\quo(\alpha)$.
  The index of~$\alpha$, denoted~$\ind(\alpha)$ is the index of its signature.
  We further impose that
  \begin{equation}\label{eq:extsigpoly}
    \poly(\alpha) - \quo(\alpha) f_{\ind (\alpha)} \in I_{\ind(\alpha) - 1}.
  \end{equation}
\end{definition}


The product of a monomial~$a\in R$ and an extended sig-poly pair~$\gamma$ is defined by
\[ \poly(a\alpha) = a \poly(\alpha), \quad  \sig(a\alpha) = a \sig(\alpha), \quad \text{and } \quo(a \alpha) = a\quo(\alpha). \]

The concept of an S-pair from Buchberger's algorithm extends to extended sig-poly pairs.
Given two extended sig-poly pairs~$\alpha$ and~$\beta$ with $\sig(\alpha) > \sig(\beta)$
let~$c = \lcm(\lm(\poly(\alpha)), \lm(\poly(\beta)))$, $a = c/\lt(\poly(\alpha))$ and~$b = c/\lt(\poly(\beta))$,
then define
the S-pair of~$\alpha$ and~$\beta$, denoted~$\sig p(\alpha, \beta)$
by
\[  \poly( \sig p(\alpha, \beta) ) = a \poly(\alpha) - b \poly(\beta), \quad \sig( \sig p(\alpha, \beta) ) = \max(\sig(a\alpha), \sig(b\beta)), \]
and
\[ \quo( \sig p(\alpha, \beta)) =
  \begin{cases}
    a \quo(\alpha) & \text{if $\ind(\alpha) > \ind(\beta)$,}\\
    a \quo(\alpha) - b\quo(\beta) & \text{if $\ind(\alpha) = \ind(\beta)$.}
  \end{cases} \]
In particular, the polynomial part of~$\sig p(\alpha, \beta)$ is the usual S-pair of~$\poly(\alpha)$ and~$\poly(\beta)$.
We say that $\alpha$ and~$\beta$ \emph{form a regular S-pair}
if~$\sig(a\alpha) \neq \sig(b\beta)$.
(We will only consider such S-pairs.)
It is easy to check that Invariant~\eqref{eq:extsigpoly} is preserved.

The \emph{regular reduction} of an extended sig-poly pair $\alpha$ with respect to a set~$G$ of sig-poly pairs is defined to be the output of Algorithm~\ref{alg:regular-reduction}.
The procedure tries to reduce the leading term of~$\poly(\alpha)$ using some multiple~$b \beta$
of an extended sig-poly pair $\beta\in G$ such that~$b\sig(\beta) < \sig(\alpha)$.
The procedure stops when there is no such reducer.
Compared to the usual division algorithm in polynomial rings, only reduction by lower signature elements is allowed. Moreover, there is some extra computations to preserve Invariant~\eqref{eq:extsigpoly}.

\begin{algorithm}[tp]
  \caption{Regular reduction}
  \label{alg:regular-reduction}
  \begin{algorithmic}[1]
    \Procedure{RegularReduction}{$\alpha$, $G$}
    \State $f \gets \poly(\alpha)$
    \State $h \gets \quo(\alpha)$
    \While{$\left\{ \beta\in G \st \lm(\poly(\beta)) \text{ divides }\lm(f) \right\} \neq \varnothing$}
    \State $\beta \gets $ some element of $\left\{ \beta\in G \st \lm(\poly(\beta)) \text{ divides }\lm(f) \right\}$
    \State $b \gets \lt(\poly(\beta)) / \lt(f)$
    \If{$b\sig(\beta) < \sig(\alpha)$}
    \State $f \gets f - b \poly(\beta)$
    \If{$\ind(\beta) = \ind(\alpha)$}
    \State $h \gets h - b \quo(\beta)$
    \EndIf
    \EndIf
    \EndWhile
    \Return $(f, \sig(\alpha), h)$
    \EndProcedure
  \end{algorithmic}
\end{algorithm}

\begin{algorithm}[tp]
  \caption{Buchberger with signatures}
  \label{alg:buchberger-signature}
  \begin{algorithmic}[1]
    \Require $f_1,\dots,f_r \in R$
    \Ensure Gröbner bases of $\langle f_1,\dots,f_r\rangle $ and of~$\langle f_1,\dotsc,f_{k-1}\rangle : f_k$ ($1\leq k\leq r$)
    \Procedure{Buchberger}{$f_1,\dotsc,f_r$}
    \State $G\gets \left\{ \epsilon_i \st 1\leq i\leq r \right\}$
    \State $S_1, \dotsc, S_r \gets \varnothing$
    \State $P\gets\left\{ (\alpha, \beta) \st \alpha, \beta\in G \text{ form a regular S-pair}  \right\}$
    \While{$P\neq \varnothing$}
    \State $(\alpha, \beta) \gets $ the pair in $P$ with $\sig(\sig p(\alpha, \beta))$ minimal
    \State $P\leftarrow P \setminus \{(\alpha, \beta)\}$
    \State $\gamma \leftarrow $ \Call{RegularReduction}{$\sig p(\alpha, \beta)$, $G$} \label{line:buch:regred}
    \State $G \leftarrow G \cup \{\gamma\}$
    \If{$\poly(\gamma) \neq 0$}
    \State $P \leftarrow P \cup \left\{ (\gamma, \beta) \st \beta\in G \text{ forms a regular S-pair with~$\gamma$} \right\}$
    \Else \quad \emph{(record the quotient of the zero reduction)}
    \State $S_{\ind(\gamma)} \gets S_{\ind(\gamma)} \cup \left\{ \quo(\gamma) \right\}$
    \EndIf
    \EndWhile
    \Return $\left\{ \poly(\beta) \st \beta\in G \right\}$, $S_1$, \ldots, $S_r$
    \EndProcedure
  \end{algorithmic}
\end{algorithm}

We may now describe a variant of Buchberger's using extended sig-poly pairs and
regular reduction, see Algorithm~\ref{alg:buchberger-signature}.
In line 6 we always choose the $S$-pair with minimal signature for reduction, and signatures are ordered first by indices.
As a result, signatures are processed in index~1 (which may produce further S-pairs with index $\geq 1$),
then in index~2 (which may produce further S-pairs with index $\geq 2$), etc.
So a Gröbner basis for $I$ is computed incrementally: first for $\langle f_1\rangle$, then for $\langle f_1,f_2\rangle$ etc.
Computing with extended sig-poly pairs makes it possible to simultaneously compute a Gröbner basis for $I$ and for all the ideals $(\langle f_1,\dots,f_{i-1}\rangle:f_i)$ for $i=2,\dots,r$. Indeed, if for an extended sig-poly pair $\gamma$ we find during the run of Algorithm \ref{alg:buchberger-signature} that $\poly(\gamma) = 0$,
then~$\quo(\gamma)$ is an element of the quotient ideal $I_{\ind(\gamma)-1} : f_{\ind(\gamma)}$,
in view of Definition~\ref{def:extended-sig-poly}.

\begin{proposition}\label{prop:buchberger-signature}
  On input~$f_1,\dotsc,f_r \in R$,
  Algorithm~\ref{alg:buchberger-signature} terminates and the set
  $\{\operatorname{poly}(\alpha)\;|\;\alpha \in G\}$
  is a Gröbner basis of the ideal~$\langle f_1,\dotsc,f_r \rangle$. The sets $S_i$ are Gröbner bases of the ideals $\langle f_1,\dotsc,f_{i-1} \rangle : f_i$ for each $i=2,\dots,r$.
\end{proposition}

We skip the proof as we will only rely on the stronger Theorem~\ref{theorem:f5}
below.

\subsection{From Buchberger to sGB}
\label{sec:from-buchberger-f5}

The signature and the quotient of each extended sig-poly pair in the
data makes it possible to compute the colon ideals $\langle f_1,\dotsc,f_{i-1} \rangle : f_i$
as a by-product of an incremental computation of a Gröbner basis of~$\langle f_1,\dotsc,f_r \rangle$.
Moreover, this is the discovery of Faugère \cite{faugere_f5},
signatures make it possible to discard many S-pairs while preserving the essential properties of Algorithm~\ref{alg:buchberger-signature}.
The overarching principle is the following: \emph{at most one sig-poly pair has to be regular-reduced at each signature}.
This is made precise by the following statement.
\begin{lemma}[{\cite[Lemma~4]{EderRoune_2013}}]\label{lem:singular-criterion}
  In the course of Algorithm \ref{alg:buchberger-signature}, assume that only $S$-pairs in signature $\geq \sigma := (i, m)$ are left in~$P$.
  Then for any extended sig-poly pairs $\gamma$ and~$\gamma'$ with~$\sig(\gamma) = \sig(\gamma') = \sigma$,
  \begin{align*}
    \textsc{RegularReduction}(\gamma, G) = \textsc{RegularReduction}(\gamma', G)
  \end{align*}
\end{lemma}

This leads to Algorithm~\ref{alg:f5syz}. It is similar to Algorithm~\ref{alg:buchberger-signature},
the only difference is the check on line~\ref{line:f5:check-rw}, the \emph{rewritability check}, which trim many computations.
At a given signature, this check will retain at most one element of~$P$.
The condition on line~\ref{line:rw:trivial-syyzgy} discards even more S-pairs by predicting that they will reduce to zero.

\begin{algorithm}[tp]
  \label{alg:rewrite}
  \caption{The rewritability criterion}
  \begin{algorithmic}[1]
    \Require $\alpha$ a sig-poly pair, $m$ a monomial, ~$G$ a set of sig-poly pairs with $\alpha \in G$
    \Ensure Returns true if~$m \alpha$ is rewritable w.r.t. $G$; false otherwise
    \Procedure{Rewritable}{$\alpha$, $m$, $G$}
    \For{$\delta \in G$}
    \If{$\sig(\delta)$ divides $\sig(m\alpha)$ and $\delta$ was added to $G$ later than $\alpha$} \label{line:rw:rewrite-order}
    \Return true \quad \emph{(Singular criterion)}
    \ElsIf{$\sig(\delta)$ divides~$\sig(m\alpha)$ and~$\poly(\delta) = 0$}
    \Return true \quad \emph{(Syzygy criterion)}
    \ElsIf{$\ind(\delta)<\ind(\alpha)$ and $\lm(\operatorname{poly}(\delta))$ divides $\lm(\quo(\alpha))$} \label{line:rw:trivial-syyzgy}
    \Return true \quad \emph{(Koszul criterion)}
    \EndIf
    \EndFor
    \Return false
    \EndProcedure
  \end{algorithmic}
\end{algorithm}

\begin{algorithm}[tp]
  \caption{sGB with recording of syzygies}
  \label{alg:f5syz}
  \begin{algorithmic}[1]
    \Require $f_1,\dotsc,f_r \in R$
    \Ensure See Theorem~\ref{theorem:f5}

    \Procedure{sGB}{$f_1,\dotsc,f_r$}
    \State $G\gets \left\{ (f_i, (i, 1), 1) \st 1\leq i\leq r \right\}$
    \State $S_1, \dotsc, S_r \gets \varnothing$
    \State $P\gets\left\{ (\alpha, \beta) \st \alpha, \beta\in G \text{ form a regular S-pair}  \right\}$
    \While{$P\neq \varnothing$}
    \State $(\alpha, \beta) \gets $ the element in $P$ with minimal signature
    \State $P\leftarrow P \setminus \{(\alpha, \beta)\}$
    \State $a,b\gets$ the monomials such that~$a \poly(\alpha) - b \poly(\beta) = \poly(\sig p(\alpha, \beta))$
    \If{not \Call{Rewritable}{$\alpha$, $a$, $G$} and not \Call{Rewritable}{$\beta$, $b$, $G$}} \label{line:f5:check-rw}
    \State $\gamma \leftarrow $ \Call{RegularReduction}{$\sig p(\alpha, \beta)$, $G$}
    \State $G \leftarrow G \cup \{\gamma\}$
    \If{$\poly(\gamma) \neq 0$}
    \State $P \leftarrow P \cup \left\{ (\gamma, \beta) \st \beta\in G \text{ forms a regular S-pair with~$\gamma$} \right\}$
    \Else \quad \emph{(record the quotient of the zero reduction)}
    \State $S_{\ind(\gamma)} \gets S_{\ind(\gamma)} \cup \left\{ \quo(\gamma) \right\}$ \label{line:f5:insertsyz}
    \EndIf
    \EndIf
    \EndWhile
    \Return $\left\{ \poly(\beta) \st \beta\in G \right\}$, $S_1$, \ldots, $S_r$
    \EndProcedure
  \end{algorithmic}
\end{algorithm}


More precisely, in the context of Lemma~\ref{lem:singular-criterion},
we can predict that all S-pairs with signature~$\gamma$ will reduce to the same element.
The first effect of the rewritability check is the discarding of all S-pairs with signature~$\sigma$, except at most one.
Secondly, Lemma~\ref{lem:singular-criterion} may be used to predict that a S-pair will reduce to zero.
There are two criteria for that:
\begin{description}
  \item[Syzygy criterion] If an element in signature~$\tau$ has reduced to zero, then every element in signature $a \tau$ (for any monomial~$a$) will reduce to zero;
  \item[Koszul criterion] If we have a sig-poly pair~$h$ with index~$< \ind(\sigma)$ and, then every element in signature~$(a \lm h, \ind (\sigma))$ will reduce to zero,
        (because~$h f_{\ind(\sigma)}$ will obviously reduce to zero).
\end{description}
This explains the different checks in the rewritability criterion (Algorithm~\ref{alg:rewrite}), see \cite[section 7.1]{eder_faugere_survey} for a detailed discussion.

\begin{theorem}\label{theorem:f5}
  On input~$f_1,\dotsc,f_r \in R$,
  Algorithm~\ref{alg:f5syz} terminates
  and outputs subsets~$G$, $S_1,\dotsc,S_r$ of~$R$ such that:
  \begin{enumerate}[(i)]
    \item $G$ is a Gröbner basis of~$\langle f_1,\dotsc,f_r \rangle$;
    \item $I_{i-1} + \langle S_i \rangle = I_{i-1} : f_i$.
  \end{enumerate}
  Moreover, on line~\ref{line:f5:insertsyz},
  when a polynomial~$g$ is inserted in some~$S_i$, then
  $\lm(g)$ is not divided by the leading monomial of any element of~$I_{i-1}$
  or any element previously inserted in~$S_i$.
\end{theorem}

\begin{proof}
  Termination and the first two
  points are a special case of \cite[Theorem~7.1]{eder_faugere_survey},
  where we only compute partial information about the syzygy module.

  The last point is a consequence from the rewritability check.
  We first note that every time a polynomial~$h$ is inserted into~$S_i$,
  the extended sig-poly pair~$(0, (i, \lm h), h)$ has been inserted into~$G$ just before.
  (The monomial part of the signature is always the leading monomial of the quotient, this is an invariant of sig-poly pairs.)
  Next, in the context of line~\ref{line:f5:insertsyz}, if~$g = \quo(\gamma)$,
  then~$\sig(\gamma) = (\ind(\gamma), \lm(g))$.
  Moreover, $\gamma$ comes from a S-pair $\sig p(\alpha, \beta)$, so~$\sig(\gamma) = a \sig(\alpha)$ or~$b\sig (\beta)$,
  and both \Call{Rewritable}{$\alpha$, $a$, $G$} and \Call{Rewritable}{$\beta$, $b$, $G$} were false.

  The Syzygy criterion implies that~$\sig(\gamma)$ is not divided by any~$\sig(\delta)$, where~$\delta\in G$ and~$\poly(\delta) = 0$.
  In other words, $\lm(g)$ is not divided by any~$\lm h$, where $h$ has been previously inserted into~$S_i$.

  The Koszul criterion implies that~$\lm(g)$ is not divided by any~$\lm(\poly(\delta))$, where~$\delta \in G$ and~$\ind(\delta) < i$.
  But due to the incremental nature of the algorithm, the set~$\left\{ \poly(\delta) \st \delta \in G, \ind(\delta) < i \right\}$
  is a Gröbner basis of~$I_{i-1}$. So~$\lm(g)$ is not divided by any element in~$I_{i-1}$.
\end{proof}





\subsection{The sGB tree datastructure}
\label{sec:f5-tree-datastr}

\subsubsection{Specification}
\label{sec:specification}

We now specify a data structure, called \emph{sGB tree}.
It is meant to extend the sGB algorithm presented above in two ways: by offering the possibility to add new input equations
during the computation; and by offering the possibility to split the computation into different branches while sharing the common base.

An sGB tree represents a rooted tree~$T$ where each node holds an element of the polynomial ring~$R$.
The nodes are partially ordered by the ancestor-descendant relation: $\nu \leq_T \mu$ if~$\nu$ is on the unique path from~$\mu$ to the root of~$T$ (or, equivalently, if~$\mu$ is in the subtree rooted at~$\nu$).
For a node~$\nu$, the polynomial contained in~$\nu$ is denoted~$\poly(\nu)$,
and the ideal generated by the polynomials contained by the ancestors of~$\nu$ (not including~$\nu$)
is denoted~$I_{< \nu}$.
An sGB tree offers the following three operations. How we implement them is the matter of the next section.
\begin{description}
  \item[Node insertion] Insert a new node, containing a given polynomial~$f$, anywhere in the tree, as a new leaf or on an existing edge.
        Denoted~\Call{InsertNode}{$\mathcal{T}$, $f$, position}.
  \item[Gröbner basis] Given a node~$\nu$, outputs a Gröbner basis of the ideal generated by the polynomials contained in the nodes~$\leq_T \nu$.
        Denoted \Call{Basis}{$\mathcal{T}$, $\nu$}.
  \item[Get a syzygy]
        Given a node~$\nu$, outputs an element of~$I_{< \nu} : \poly(\nu)$.
        Denoted \Call{GetSyzygy}{$\mathcal{T}$, $\nu$}.

        If \Call{GetSyzygy}{$\mathcal{T}$, $\nu$} outputs zero,
        then~$I_{<\nu} + J = I_{< \nu} : \poly(\nu)$, where~$J$ is the ideal generated by
        all previous invocations of~\Call{GetSyzygy}{$\mathcal{T}$, $\nu$}.

        It is guaranteed that \Call{GetSyzygy}{$\mathcal{T}$, $\nu$} eventually outputs zero
        after sufficiently many invocation,
        even if nodes are inserted or \textsc{GetSyzygy} is called on other nodes in between.
\end{description}

\subsubsection{Implementation}
\label{sec:implementation}

\begin{algorithm}[tp]
  \caption{Implementation of the sGB tree data structure}
  \label{alg:sgbtree}
  \begin{algorithmic}[1]
    \Require An sGB tree~$\mathcal{T}$ and a label of~$T$
    \Ensure Process the pair in~$P$ with index above~$\nu$ with smallest signature
    \Procedure{ProcessSPair}{$\mathcal{T}$, $\nu$}
    \State \emph{(restrict to S-pairs whose indices are above~$\nu$)}
    \State $P' \gets \left\{ (\alpha, \beta) \st \max \left\{ \ind(\alpha), \ind(\beta) \right\} \leq_T \nu \right\}$
    \If{$P' \neq \varnothing$}
    \State $(\alpha, \beta) \gets$ the pair in~$P'$ with $\sig(\sig p(\alpha, \beta))$ minimal
    \State $P \gets P \setminus \left\{ (\alpha, \beta) \right\}$
    \State $a,b\gets$ the monomials such that~$a \poly(\alpha) - b \poly(\beta) = \poly(\sig p(\alpha, \beta))$
    \If{not \Call{Rewritable}{$\alpha$, $a$, $G$} and not \Call{Rewritable}{$\beta$, $b$, $G$}} 
    \State $\gamma \leftarrow $ \Call{RegularReduction}{$\sig p(\alpha, \beta)$, $G$}
    \State $G \gets G \cup \left\{ \gamma \right\}$
    \If{$\poly(g) \neq 0$}
    \State $P \leftarrow P \cup \left\{ (\gamma, \beta) \st \beta\in G \text{ forms a regular S-pair with~$\gamma$} \right\}$
    \Else \quad \emph{(record the quotient of the zero reduction)}
    \State $S_{\ind(\gamma)} \gets S_{\ind(\gamma)} \cup \left\{ \quo(\gamma) \right\}$ \label{line:processspair:newsyz}
    \EndIf
    \EndIf
    \EndIf
    \EndProcedure
  \end{algorithmic}

  \medskip
  \begin{algorithmic}[1]
    \Require A sGB tree~$\mathcal{T}$ and a label~$\nu$ of~$T$
    \Ensure A Gröbner basis of~$I_{<\nu}$
    \Procedure{Basis}{$\mathcal{T}$, $\nu$}
    \While{there is a pair in~$P$ with index $\leq_T \nu$}
    \State \Call{ProcessSPair}{$\mathcal{T}$, $\nu$}
    \EndWhile
    \Return $\left\{ \poly(\alpha) \st \alpha \in G \text{ and } \ind(\alpha) \leq_T \nu \right\}$
    \EndProcedure
  \end{algorithmic}

  \medskip
  \begin{algorithmic}[1]
    \Require A sGB tree~$\mathcal{T}$ and a label~$\nu$ of~$T$
    \Ensure  An element of the quotient ideal~$I_{<\nu} : \poly(\nu)$ not contained in $I_{<\nu}$
    \Procedure{GetSyzygy}{$\mathcal{T}$, $\nu$}
    \While{there is a pair in~$P$ with index $\leq_T\nu$ \textbf{and} $S_\nu = \varnothing$}
    \State \Call{ProcessSPair}{$\mathcal{T}$, $\nu$}
    \EndWhile
    \If{$S_\nu \neq \varnothing$}
    \State pick and remove some~$h$ in~$S_\nu$
    \Return $h$
    \Else
    \Return 0
    \EndIf
    \EndProcedure
  \end{algorithmic}
\end{algorithm}

From the point of implementation,  an sGB tree
is made of:
\begin{enumerate}[(a)]
  \item a rooted tree $T$ containing whose nodes are labelled with integers;
  \item a set~$G$ of extended sig-poly pairs whose indices are nodes of~$T$ (see below);
  \item a set~$P$ of pairs of elements of~$G$ forming regular S-pairs;
  \item for each node~$\nu$ of~$T$, a subset~$S_\nu$ of~$R$.
\end{enumerate}

The sets~$G$, $P$ and~$S_\nu$ have the same role as their counterparts in the sGB algorithm (Algorithm~\ref{alg:f5syz}).
The main difference is a twist in the definition of signatures and indices.
In~\S \ref{sec:sign-extend-sig}, an index (that is the first component of a signature)
is a nonnegative integer.
From now on, indices are nodes in~$T$.
Indices are partially ordered by the ancestor-descendant relation~$\leq_T$.
Note that for a given node~$\nu$, the subset $\left\{ \mu \st \mu \leq_T \nu \right\}$ is totally ordered: it is the set of nodes on the path from the root of~$T$ to~$\nu$.
Lastly, we adjust the definition of a regular~S-pair.
We say that sig-poly pairs~$\alpha$ and~$\beta$ form a regular S-pair if~$\ind(\alpha)$ and~$\ind(\beta)$ are comparable (that is either~$\ind(\alpha) \leq_T \ind(\beta)$ or~$\ind(\beta) \leq_T \ind(\alpha)$)
and $\sig(a\alpha) \neq \sig(b\beta)$, with~$a$ and~$b$ as in~\S \ref{sec:sign-extend-sig}.
To analyze the behavior of the sGB-tree data structure, we always consider totally ordered subsets of indices,
thus reducing to the context of Algorithm \ref{alg:f5syz}.

To implement \Call{Basis}{$\mathcal{T}$, $\nu$},
we process the S-pairs with index~$\leq_T \nu$.
The indices of these S-pairs are totally ordered, so we are actually in the situation of~\S \ref{sec:from-buchberger-f5}
and we may apply the main loop of Algorithm~\ref{alg:f5syz}.
The body of this loop is isolated in procedure \textsc{ProcessSPair} (Algorithm~\ref{alg:sgbtree}), with the appropriate alterations.

The implementation of~\Call{GetSyzygy}{$\mathcal{T}$, $\nu$} is similar, with the difference that we abort the computation as soon as the set~$S_\nu$ is not empty and return an element of it, see Algorithm~\ref{alg:sgbtree}. If~$S_\nu$ is still empty after having processed all S-pairs which may lead to new elements in~$S_\nu$, the value~$0$ is returned.

We assume that the state of a sGB tree always results from a sequence of calls to \textsc{InsertNode}, \textsc{Basis} or \textsc{GetSyzygy}
applied to an initially empty tree.

\begin{algorithm}[tp]
  \caption{The sGB tree data structure, insertion of a node}
  \label{alg:insertnode}
  \begin{algorithmic}[1]

    \Require A sGB tree~$\mathcal{T}$, a polynomial~$f$ and a description of the position of the new node in~$T$
    \Ensure The label of the newly inserted node
    \Procedure{InsertNode}{$\mathcal{T}$, $f$, position}
    \State $\nu \gets \text{(largest label in~$T$)} + 1$
    \State insert a node in~$T$ with label~$\nu$, as described by ``position''
    \State $S_\nu \gets \varnothing$
    \State $\epsilon \gets (f, (\nu, 1), 1)$
    \State $P \gets \left\{ (\epsilon, \beta) \st \beta \in G \text{ and } (\epsilon, \beta) \text{ is regular} \right\}$
    \State $G \gets G \cup \left\{ \epsilon \right\}$
    \Return $\nu$
    \EndProcedure
  \end{algorithmic}
\end{algorithm}


\begin{proposition}
  Let~$\mathcal{T}$ be a sGB tree and let~$\nu$ be a node of~$\mathcal{T}$.
  {\upshape \Call{Basis}{$\mathcal{T}$, $\nu$}}
  (Algorithm~\ref{alg:sgbtree}) terminates and
  outputs a Gröbner basis of~$I_\nu$.
\end{proposition}

\begin{proof}
  This algorithm considers only S-pairs whose signatures are above a given node~$\nu$.
  After this restriction, the signature are totally ordered, so \textsc{Basis} behaves exacly like Algorithm~\ref{alg:f5syz} (\textsc{sGB}).
  We note that, contrary to \textsc{sGB},
  \textsc{Basis}
  may start in a state where several S-pairs have already been processed, in an unspecified order, by earlier calls to \textsc{Basis} or \textsc{GetSyzygy}
  on different nodes.
  This does not invalidate neither the termination proof given in \cite{EderRoune_2013}, nor the proof of correctness.
\end{proof}

\begin{proposition}\label{prop:getsyzygy}
  Let~$\mathcal{T}$ be a sGB tree and let~$\nu$ be a node of~$\mathcal{T}$.
  {\upshape \Call{GetSyzygy}{$\mathcal{T}$, $\nu$}}
  (Algorithm~\ref{alg:sgbtree}) terminates and
  outputs some~$f \in R$ such that:
  \begin{enumerate}[(i)]
    \item\label{item:getsyz:correct} $f \in I_{<\nu} : \poly(\nu)$;
    \item\label{item:getsyz:nonredundant} if~$f\neq 0$, then $\lm(f)$ is not divisible by the leading monomial of any other polynomial previously output by   {\upshape \Call{GetSyzygy}{$\mathcal{T}$, $\nu$}}, or any polynomial in~$I_{< \nu}$;
    \item\label{itme:getsyz:zero} if~$f = 0$, then $I_{<\nu}:\poly(\nu)$ is generated by~$I_{<\nu}$ and the polynomials previously output by  {\upshape \Call{GetSyzygy}{$\mathcal{T}$, $\nu$}}.
  \end{enumerate}
\end{proposition}

\begin{proof}
  Termination follows from the termination of \textsc{Basis} since the main loop is similar, but with the possibility of earlier termination.
  Correctness follows from Theorem~\ref{theorem:f5} after restricting to indices above~$\nu$.
\end{proof}

As a consequence of Proposition~\ref{prop:getsyzygy}\ref{item:getsyz:nonredundant}, it is guaranteed that
\Call{GetSyzygy}{$\mathcal{T}$, $\nu$} eventually outputs zero after
sufficiently many invocation, even if nodes are inserted or \textsc{GetSyzygy}
is called on other nodes in between. Indeed, the leading monomial of a nonzero
output of \Call{GetSyzygy}{$\mathcal{T}$, $\nu$} is constrained to be outside
the monomial ideal generated by the leading monomials of previous output.
By Dickson's lemma, this may only happen finitely many times.




\section{Computation of the nondegenerate locus}

The sGB-tree data structure may can be used to implement an efficient variant of Algorithm~\ref{alg:decomp} for computing the nondegenerate locus.
We use a sGB tree to compute efficiently saturations~$I : f^\infty$, and also double quotient~$I : (I : f^\infty)$,
with the idea to exploit as soon as possible newly discovered relations to simplify further computations.
This leads to Algorithm~\ref{alg:nondeg-f5tree}, which we describe informaly as follows.

Similarly to Algorithm~\ref{alg:decomp},
we introduce the equations~$f_1,\dotsc,f_r$ one after the other.
We maintain a sGB tree which, at the beginning of the~$k$th iteration, that is after having processed~$f_1,\dotsc,f_{k-1}$, has the following shape:
\[ \mathbf{g}_1 \gets f_1 \gets \mathbf{p}_1 \gets \dotsb \gets \mathbf{g}_{k-1} \gets f_{k-1} \gets \mathbf{p}_{k-1} \gets \underbrace{0}_{\nu}
  \raisebox{-.3em}{$\begin{matrix}
  \swarrow h_1\\
  \gets h_2 \\
  \phantom{\gets} \vdots
\end{matrix}$},
\]
where bold letters represent a sequence of zero, one or several nodes.
The tree grows from the node labeled~$\nu$, by adding new leaf nodes, or inserting nodes just above~$\nu$.
Using the notations of Algorithm~\ref{alg:decomp},
the nodes~$\mathbf{g}_i$ are related to the saturation~$G: f_i^\infty$,
the leaf nodes~$h_i$ are generic elements of the ideals in the set~$\mathcal{K}$,
and the nodes~$\mathbf{p}_i$ are related to the cleaning steps~$G : K^\infty$. The leaf nodes $h_i$ are generic
in the sense that they are either each a random linear combination of generators of the ideals in $\mathcal{K}$
or each a linear combination of of generators of the ideals in $\mathcal{K}$ with each coefficient a new variable.

The~$k$th iteration proceeds as follows.
Firstly, a new node~$\mu$ containing $f_k$ is created just above $\nu$:
\[ \dotsb \gets \underbrace{f_k}_{\mu} \gets \underbrace{0}_{\nu} \gets \dotsb. \]
As long as \Call{GetSyzygy}{$\mathcal{T}$, $\mu$} returns nonzero elements ($g_1,g_2,\dotsc$),
we insert them above~$\mu$:
\[ \dotsb \gets g_1 \gets g_2 \gets \dotsb \gets \underbrace{f_k}_{\mu} \gets \underbrace{0}_{\nu} \gets \dotsb. \]
This saturation has the effect of completing~$I_{< \mu}$ into~$I_{< \mu} : f_k^\infty$.
Each time we insert a polynomial~$g_i$ in a node, say~$\gamma$, we also record the syzygies \Call{GetSyzygy}{$\mathcal{T}$, $\gamma$},
take a generic linear combination and insert it as a new leaf node. These syzygies are related to the double quotient~$I_{< \mu}:(I_{<\mu} : f_k^\infty)$.
Before going to the next iteration, insert above~$\nu$ all the syzygies obtained from the children of~$\nu$.
Which again has the effect of saturating~$I_{<\nu}$ by the polynomials contained in these nodes.

After all the input equations have been processed, the ideal~$I_{<\nu}$ is a nondegenerate part of the input ideal,
which we prove by comparing with Algorithm~\ref{alg:decomp}.

\begin{algorithm}[tp]
  \caption{Computation of the nondegenerate locus with an sGB tree}
  \label{alg:nondeg-f5tree}
  \begin{algorithmic}[1]
    \Require $f_1,\dots,f_c \in R$
    \Ensure A Gröbner basis $G$ of a nondegenerate locus of $(f_1,\dots,f_c)$
    \State $\mathcal{T} \gets$ an empty sGB tree
    \State $\nu \gets $ \Call{InsertNode}{$\mathcal{T}$, $0$}
    \For{$k$ from~$1$ to~$c$}
    \State $\mu \gets$ \Call{InsertNode}{$\mathcal{T}$, $f_k$, just above~$\nu$} \label{line:insertmu}
    \Loop  \label{line:whilemu}
      \State $g \gets$ \Call{GetSyzygy}{$\mathcal{T}$, $\mu$} \label{line:getsyzmu}
      \If{$g = 0$} \State \textbf{break} \EndIf
      \State $P \gets P \cup \{(\poly(\mu), (\mu, 1), 1)\}$
    \State $\gamma \gets $ \Call{InsertNode}{$\mathcal{T}$, $g$, just above~$\mu$} \label{line:insertg}
    \State $h\gets 0$
    \State $t \gets $ a random scalar (or the slack variable, see Remark~\ref{remark:deterministic}) \label{line:randomscalar}
    \Loop  \label{line:whilegamma}
    \State $h' \gets$ \Call{GetSyzygy}{$\mathcal{T}$, $\gamma$}
    \If{$h'=0$} \State \textbf{break} \EndIf
    \State $h \gets t h + h'$
    \EndLoop
    \State \Call{InsertNode}{$\mathcal{T}$, $h$, as a child of~$\nu$} \label{line:inserth}
    \EndLoop
    \ForAll{child~$\beta$ of~$\nu$} \label{line:loopcleaning}
    \Loop \label{line:loopcleaning-inner}
    \State $b \gets$ \Call{GetSyzygy}{$\mathcal{T}$, $\beta$}
    \If{$b=0$} \State \textbf{break} \EndIf
    \State \Call{InsertNode}{$\mathcal{T}$, $b$, just above~$\nu$}
    \EndLoop
    \EndFor
    \EndFor
    \Return $\operatorname{Basis}(\nu)$
  \end{algorithmic}
\end{algorithm}

\begin{theorem}
  Algorithm~\ref{alg:nondeg-f5tree} terminates and is correct.
\end{theorem}

\begin{proof}
  Termination follows from the assumption that for any node~$\nu$ of an sGB tree~$\mathcal{T}$,
  \Call{GetSyzygy}{$\mathcal{T}$, $\nu$} eventually returns~0 after sufficiently many calls.

  To prove correctness, we show that Algorithm~\ref{alg:nondeg-f5tree} computes the same ideal as Algorithm~\ref{alg:decomp}.
  Let~$J_{k-1}$ be the value of~$I_\nu$ at the beginning of the $k$th iteration.
  After line~\ref{line:insertmu}, we also have~$I_{<\mu} = J_{k-1}$, while~$I_{<\nu} = I_{<\mu} + \langle f_k \rangle$.

  We first examine the loop on line~\ref{line:whilemu}.
  It inserts above the node~$\mu$ all the polynomials obtained from \Call{GetSyzygy}{$\mathcal{T}$, $\mu$}.
  Every node inserted on line~\ref{line:insertg} is in~$I_{<\mu} : f_k$. No other node is inserted above~$\mu$. So by induction, it follows
  that all along the loop, we have~$I_{<\mu} \subseteq J_{k-1} : f_k^\infty$.
  Moreover, after the loop terminates, we have~$I_{<\mu} : f_k = I_{<\mu}$, due to the specification of \textsc{GetSyzygy} (Proposition~\ref{prop:getsyzygy}).
  It follows that before line~\ref{line:loopcleaning}, we have
  \begin{equation}\label{eq:1}
    I_{<\mu} = J_{k-1} : f_k^\infty \quad \text{and} \quad I_{<\nu} = ( J_{k-1} : f_k^\infty) + \langle f_k \rangle.
  \end{equation}

  Next, we examine the loop on line~\ref{line:loopcleaning} and its inner loop on line~\ref{line:loopcleaning-inner}.
  By the same argument as above, the inner loop has the effect of saturating~$I_{<\nu}$ by $\operatorname{pol}(\beta)$.
  So after the loop on line~\ref{line:loopcleaning}, we have
  \begin{equation}
    I_{<\nu} = J_k = \left( ( J_{k-1} : f_k^\infty) + \langle f_k \rangle \right) : \bigg( \prod_{\beta \text{ child of } \nu} \operatorname{pol}(\beta) \bigg)^\infty.\label{eq:3}
  \end{equation}

  It remains to understand the nature of the children of~$\nu$.
  They all come from the insertion of~$h$ on line~\ref{line:inserth}.
  And~$h$ is simply a generic linear combination of the return values of \Call{GetSyzygy}{$\mathcal{T}$, $\gamma$}.
  So~$h$ is a generic linear combination of some~$h_1,\dotsc,h_r$ such that~$I_{<\gamma} + \langle h_1,\dotsc,h_r\rangle = I_{<\gamma} : \poly(\gamma)$ (by Proposition~\ref{prop:getsyzygy}).
  For each node~$\gamma$ inserted on line~\ref{line:insertg}, let~$L_\gamma$ denote the ideal~$I_{<\gamma} : \poly(\gamma)$.
  If~$g_1,\dotsc,g_s$ are the successive return values of \Call{GetSyzygy}{$\mathcal{T}$, $\mu$} on line~\ref{line:getsyzmu},
  and~$\gamma_1,\dotsc,\gamma_r$ the corresponding nodes,
  we have~$L_{<\gamma_{i}} = I_{<\gamma_i} : g_i$
  and~$I_{<\gamma_i} = J_{k-1} + \langle g_1,\dotsc,g_i \rangle$.
  By Lemma~\ref{lem:lazy_rep}, it follows that
  \begin{equation}\label{eq:2}
    L_{\gamma_1} \cap \dotsb \cap L_{\gamma_r} \radeq J_{k-1} : \langle g_1,\dotsc,g_r \rangle^\infty.
  \end{equation}
  Moreover, by \eqref{eq:1}, we obtain that before line~\ref{line:loopcleaning}
  \begin{equation}
    I_{<\mu} = J_{k-1} + \langle g_1,\dotsc,g_r \rangle = J_{k-1}:f_k^\infty,
  \end{equation}
  so, combining with \eqref{eq:2},
  \begin{align}
    L_{\gamma_1} \cap \dotsb \cap L_{\gamma_r} &\radeq  J_{k-1} : \langle g_1,\dotsc,g_r \rangle^\infty \\
                                               &=  J_{k-1} : \big( J_{k-1} + \langle g_1,\dotsc,g_r \rangle\big)^\infty \\
                                               &\radeq J_{k-1} : (J_{k-1}:f_k^\infty). \label{eq:5}
  \end{align}

  As remarked above, the loop on line~\ref{line:loopcleaning-inner}
  has the effect of saturating~$I_{<\nu}$ by $\operatorname{pol}(\beta)$.
  By the analysis above, $\operatorname{pol}(\beta)$ is actually a generic linear combination of
  some $h_1,\dotsc,h_r$ such that~$I_{<\gamma} + \langle h_1,\dotsc,h_r\rangle = L_\gamma$,
  for some node~$\gamma$ above~$\nu$.
  By Lemma~\ref{prop:gen_colon},
  saturating by~$\operatorname{pol}(\beta)$ is the same as saturating by~$\langle h_1,\dotsc,h_r\rangle$.
  Besides, $I_{<\nu}$ contains~$I_{<\gamma}$,
  so saturating~$I_{<\nu}$ by~$\langle h_1,\dotsc,h_r\rangle$ is the same as saturating by~$L_\gamma$.
  Back to~\eqref{eq:3}, we conclude from~\eqref{eq:5} that saturating~$I_{<\nu}$ by all the~$\operatorname{pol}(\beta)$
  is the same as saturating by all the ideals~$J_{i-1}:(J_{i-1} : f_i^\infty)$, for~$i \leq k$.

  Therefore~$J_k$ satisfies the same recurrence relation as its analogue defined the proof of Theorem~\ref{thm:correctness}:
  \begin{equation}
    J_k = \bigg( ( J_{k-1} : f_k^\infty) + \langle f_k \rangle \bigg) : \bigg( \bigcap_{i\leq k} \big(J_{i-1}:(J_{i-1} : f_i^\infty)\big) \bigg)^\infty.
  \end{equation}
  This proves that Algorithm~\ref{alg:nondeg-f5tree} and Algorithm~\ref{alg:decomp} compute the same ideal.
\end{proof}

\begin{lemma}
  \label{lem:lazy_rep}
  Let $I,J\subseteq R$ be two ideals and let $J=\langle g_1,\dots,g_t\rangle$.
  Then
  \begin{align*}
    (I:J) \radeq (I:g_1) \cap ((I + \langle g_1\rangle) :g_2) \cap \dots \cap ((I + \langle g_1,\dots,g_{t-1} \rangle) : g_t).
  \end{align*}
\end{lemma}

\begin{proof}
  The inclusion ''$\subseteq$'' is obvious. Now, let $p\in R$ be such that
  \begin{align*}
    p^m\in (I:g_1) \cap ((I + \langle g_1\rangle) :g_2) \cap \dots \cap ((I + \langle g_1,\dots,g_{t-1} \rangle) : g_t)
  \end{align*}
  for some $m\in \mathbb{N}$. Then we have in particular $p^mg_1\in I$. Now let
  $i > 1$. By induction, if for some $k\in \mathbb{N}$ we have $p^kg_j\in I$ for
  all $j\leq i$ then
  \begin{align*}
    p^{km}g_{i+1} = p^kf + p^ka_1g_1 + \dots + p^ka_ig_i \in I
  \end{align*}
  for a suitable $f\in I$, $a_1,\dots,a_i\in R$ and so $p^{km}\in (I:g_{i+1})$.
  We deduce that a power of $p$ actually lies in $(I:J)$ which ends the proof.
\end{proof}

\begin{lemma}
  \label{prop:gen_colon}
  Let $I,J\subseteq R$ be two ideals with $J = \langle g_1,\dots,g_t\rangle$.
  \begin{enumerate}
  \item There exists a Zarisiki-open subset $D\subset \mathbb{K}^t$ such that for any $(a_1,\dots,a_t)\in D$ we have $\sat{I}{J} = \sat{I}{(\sum_{j=1}^ta_jg_j)}$.
  \item If $K\radeq J$ then $\sat{I}{K} \radeq \sat{I}{J}$.
  \end{enumerate}
\end{lemma}

\begin{proof}
  (1) easily follows e.g. from \cite[Exercise 15.41]{eisenbud_comm-alg}. For (2),
  if $p\in R$ such that $p^kJ^l\in I$ for $k,l\in \mathbb{N}$ then for a suitably
  large $m\in \mathbb{N}$ we have $K^m\subseteq J^l$ so $p^kK^m\in I$ and hence
  $p\in \sqrt{\sat{I}{K}}$.
\end{proof}

\begin{remark}[Deterministic variant]\label{remark:deterministic}
  The cleaning steps in Algorithm~\ref{alg:nondeg-f5tree} can be made in a randomized way, with a possibility of undetected error,
  or in a deterministic way. The only change to operate is on line~\ref{line:randomscalar}.
  For a randomized algorithm, favoring speed over certain correctness, choose~$t$ to be a random scalar.
  For a deterministic algorithm, choose~$t$ to be a slack variable, unused in the input equation. It is guaranteed that such a~$t$ is generic enough.
  Whenever we introduce such a slack variable we can extend the monomial ordering on $R$ in any way we like, since all cofactors of syzygies that are
  inserted as new nodes only involve the variables of $R$. The implementation discussed in the next section exclusively chooses $t$ to be a random
  scalar.
\end{remark}

\section{Implementation and Experiments}


\subsection{Further Implementational Considerations}
\label{sec:opt_sat}

We start by describing some further optimizations in our implementations of Algorithms~\ref{alg:f5syz} and \ref{alg:nondeg-f5tree}.

Both these implementations use an \emph{F4-like reduction strategy}. This means that several $S$-pairs are selected out of the pairset at once and are subsequently, together with their regular reducers, organized in a matrix whose rows are labeled by the selected extended sig-poly pairs and whose columns are labeled by all the monomials occuring in the polynomial parts of these extended sig-poly pairs. This matrix is then put into row echelon form and the rows of this reduced matrix whose first entry has changed during the computation of this row echelon form are then processed as new basis elements or newly identified zero divisors, depending on if this reduced row is zero or not. We refer to \cite{faugere_f4} for the original F4 algorithm or to \cite[section 13]{eder_faugere_survey} for an explanation as to how to combine the F4 algorithm with signature-based techniques.

For Algorithm \ref{alg:nondeg-f5tree}, this has the consequence that the \textsc{GetSyzygy} routine has the ability to return several zero divisors $g_1,\dots,g_s$ at once and Algorithm \ref{alg:nondeg-f5tree}
may benefit from it.
We implemented the following probabilistic optimization: We replaced $g_1$ by a random linear combination $g_1':= \sum_{j=1}^sa_ig_i$. Let $\nu_1,\dots,\nu_s$ be the nodes assigned to $g_1',g_{2},\dots,g_s$ in Algorithm \ref{alg:nondeg-f5tree}. Then, if the choice of the $a_i$ was ``sufficiently random'', we know by Lemma \ref{prop:gen_colon} that for $h\in R$ we have
\begin{align*}
  hg_1'\in I_{< \nu_1}\quad \Leftrightarrow \quad hg_i\in I_{< \nu_1}\forall i.
\end{align*}
If then \textsc{GetSyzygy}$(\mathcal{T}, \nu_1)$ returned such an element $h\neq 0$ we regarded the signatures $(\nu_2,\lm(h)),\dots,(\nu_s,\lm(h))$ as known signatures of syzygies during the calls to \textsc{Rewriteable}, i.e. \textsc{GetSyzygy}$(\mathcal{T}, \nu_i)$ would, for $i=2,\dots,s$, only return a non-zero result if there exists an element $h'\in (I_{<\nu_i}:g_i)$ with $\lm(h')$ not divisible by $\lm(h)$. Furthermore, only the zero divisors of $h$ of $g_1'$ as above were considered in the loop from line 14-20 of Algorithm \ref{alg:nondeg-f5tree}.

We implemented both Algorithm \ref{alg:f5syz} and \ref{alg:nondeg-f5tree} in the
programming language {Julia} \cite{julia} with an interface to the
{Singular.jl} {Julia}-library \cite{Singular_CAS}. An interface to the
new computer algebra system {OSCAR} \cite{OSCAR} is planned for the
future. The implementation is available at

\begin{center}
  \url{https://github.com/RafaelDavidMohr/SignatureGB.jl}
\end{center}

In this implementation we use our own data structures for polynomials and polynomial arithmetic.
The linear algebra routines for computing row echelon forms in our implementations closely follow the corresponding routines presented in
\cite{monagan_f4_implementation}. Additionally, our implementation makes use of the modifications
to Algorithm \ref{alg:f5syz} presented in \cite{eder_perry_f5c}. Currently the implementation
works only for fields of finite characteristic.

While our implementation is currently not competitive with optimized implementations
of Gröbner basis algorithms such as in {Maple} \cite{maple} or {msolve}
\cite{msolve}, we do make use of some standard optimization techniques in Gröbner basis algorithm
implementations such as monomial hash tables and divisor bitmasks (see e.g. \cite{roune_practical}
for a description of these techniques).

\subsection{Experimental Results}
\label{sec:experiments}

We used the following examples to benchmark our implementations:
\begin{enumerate}
\item Cyclic$(8)$, coming from the classical Cyclic$(n)$ benchmark.
\item Pseudo$(n)$, encoding pseudo-singularities as follows
  $$f_1=\cdots=f_{n-1}=g_1\cdots=g_{n-1}$$
  with $f_i\in \mathbb{K}[x_1, \ldots, x_{n-2}, z_1, z_2]$, $f_i\in
  \mathbb{K}[y_1, \ldots, y_{n-2}, z_1, z_2]$, $f_i$ being chosen as a random
  dense quadric and $g_i$ equalling $f_i$ when substituting $y_1, \ldots,
  y_{n-2}$ by $x_1, \ldots, x_{n-2}$.
\item Sos$(s, n)$, encoding the critical points of the restriction of the
  projection on the first coordinate to a hypersurface which is a sum of $s$
  random dense quadrics in $\mathbb{K}[x_1, \ldots,
  x_n]$. 
  $$f, \frac{\partial f}{\partial x_2}, \ldots, \frac{\partial f}{\partial x_n},
  \quad f = \sum_{i=1}^s g_i^2.$$
\item Sing$(n)$, encoding the critical points of the restriction of the
  projection on the first coordinate to a (generically singular) hypersurface
  which is defined by the resultant of two random dense quadrics $A, B$ in
  $\mathbb{K}[x_1, \ldots, x_{n+1}]$:
  $$f, \frac{\partial f}{\partial x_2}, \ldots, \frac{\partial f}{\partial x_n},
  \quad f = \textrm{resultant}(A, B, x_{n+1}).$$
\item The Steiner polynomial system, coming from \cite{BST20}.
\end{enumerate}
All these systems are generated by a number of polynomials equal to the
number of variables of the underlying polynomial ring. They all have
components of different dimensions, one of those being zero-dimensional,
i.e. they have a nontrivial nondegenerate locus.

In Table \ref{table:compare} we compare Algorithm \ref{alg:nondeg-f5tree} and a straightforward
implementation of ours of Algorithm \ref{alg:decomp} in {Maple}. In this implementation,
we saturated an ideal $J$ by an ideal $K$ by picking a random linear combination $p$ of generators
of $K$ and saturating $J$ by $p$ using {Maple}'s internal saturation routine.
Table \ref{table:compare} shows the improvement of Algorithm \ref{alg:nondeg-f5tree} over Algorithm
\ref{alg:decomp}: While {Maple}'s Gröbner basis engine beats our implementation of
Algorithm \ref{alg:f5syz} by a wide margin the ratio between the timings of our F5 implementation
and our implementation of Algorithm \ref{alg:nondeg-f5tree} is much better than the ratio between
the time it took to compute a Gröbner basis in {Maple} and our {Maple} implementation of Algorithm \ref{alg:decomp}. This can be seen by looking at the two respect ``ratio''-columns of table \ref{table:compare}. To additionally show the overhead of Algorithm \ref{alg:nondeg-f5tree} over Algorithm \ref{alg:f5syz} we noted the number of arithmetic operations in $\mathbb{K}$ when running each of the two algorithms on the polynomial system in question. Our implementation of Algorithm \ref{alg:nondeg-f5tree} never takes more than 10 times the number of arithmetic operations Algorithm \ref{alg:f5syz} takes, on certain examples we compare very favorably in terms of arithmetic operations to Algorithm \ref{alg:f5syz}.
\looseness=-1

In Table \ref{table:other} we compare Algorithm \ref{alg:nondeg-f5tree} to other ideal decomposition
methods available in the computer algebra systems {Singular}, {Maple} and
{Macaulay2} \cite{M2}. In {Singular} there is an elimination method~\cite{decker_greuel_pfister_survey} and an implementation of the algorithm for equidimensional decomposition presented in \cite{EisenbudHunekeVasconcelos_1992}. In {Maple} we compared against the Regular Chains
package \cite{chen_maza_triangular, CLMPX07}. In {Macaulay2} one is able to compute the
intersection of all components of non-minimal dimension again with the method presented in \cite{EisenbudHunekeVasconcelos_1992}. We then saturated the original ideal by
the result to obtain the nondegenerate locus. On a high level, our algorithm works similarly, incrementally obtaining information
about the component of higher dimension and then removing it via saturation. One should keep
in mind that all of these methods, compared to Algorithm \ref{alg:nondeg-f5tree}, work more generally:
Except for what we tried in {Macaulay2} they are all able to obtain a full equidimensional
decomposition of the input ideal.

We gave all of these methods at least an hour for each polynomial system and at most roughly 50 times the time our implementation of Algorithm \ref{alg:nondeg-f5tree} took. We indicated when these times were exceeded by using ''>'' in Table \ref{table:other}. We computed all examples on a single Intel Xeon Gold 6244 CPU @ 3.60GHz with a limit of 200G memory. If this limit was exceeded, or if another segfault occured, we indicate it with 'segfault' in Table \ref{table:other}.

\newpage
  \centering
\begin{table*}[hbt!]
  \caption{Comparing Algorithm \ref{alg:decomp} and Algorithm \ref{alg:nondeg-f5tree}}
  \label{table:compare}
  \begin{minipage}{1.0\linewidth}
        \begin{adjustwidth}{-2cm}{-1cm}
    \begin{center}
      \scriptsize
      \begin{tabular}{lllllllll}
        \toprule
        & Alg. \ref{alg:f5syz} arith. op. & Alg. \ref{alg:nondeg-f5tree} arith. op. & Alg. \ref{alg:f5syz} & Alg. \ref{alg:nondeg-f5tree} & Ratio & GB in Maple & Alg. \ref{alg:decomp} in Maple & Ratio\\
\midrule
Cyclic 8 & $1.2\cdot 10^{10}$ & $1.3\cdot 10^{11}$ & 4m & 40m & 10 & 1.2s & 154m & 7700\\
Pseudo(2, 12) & $5.3\cdot 10^{7}$ & $3.1\cdot 10^{8}$ & 1.16s & 5.2s & 4.5 & 0.268s & 3.44s & 13\\
Sing(2, 10) & $5.6\cdot 10^{7}$ & $6.5\cdot 10^{7}$ & 1.9s & 2.9s & 1.5 & 0.11s & 1.642s & 14.5\\
Sing(2, 9) & $2.5\cdot 10^{7}$ & $2.9\cdot 10^{7}$ & 1.1s & 1.4s & 1.27 & 0.06s & 0.788s & 13.1\\
Sos(2,5,4) & $1.3\cdot 10^{8}$ & $1.1\cdot 10^{8}$ & 8.5s & 7.3s & 0.85 & 0.022s & 0.479s & 21.3\\
Sos(2,6,3) & $2.1\cdot 10^{7}$ & $2.1\cdot 10^{7}$ & 1.11s & 1.4s & 1.26 & 0.021s & 0.261s & 12.4\\
Sos(2,6,4) & $4.8\cdot 10^{9}$ & $3.8\cdot 10^{9}$ & 148s & 169s & 1.14 & 0.172s & 22.7s & 132\\
Sos(2,6,5) & $4.2\cdot 10^{9}$ & $2.0\cdot 10^{9}$ & 75s & 43s & 0.57 & 0.458s & 10.38s & 22.7\\
Sos(2,7,3) & $1.3\cdot 10^{8}$ & $6.7\cdot 10^{8}$ & 5.2s & 41s & 7.9 & 0.047s & 7.162s & 152.4\\
Sos(2,7,4) & $6.5\cdot 10^{9}$ & $4.5\cdot 10^{10}$ & 3m & 32m & 10.7 & 0.433s & 1h & 8314\\
Sos(2,7,5) & $7.2\cdot 10^{10}$ & $3.5\cdot 10^{11}$ & 25m & 20h & 48 & 2.294s & >359h & $>4.4\cdot10^6$\\
Sos(2,7,6) & $1.7\cdot 10^{12}$ & $3.0\cdot 10^{12}$ & 31h & 73h & 2.4 & 14.348s & 5.5h & 23\\
        Steiner & $3.1\cdot 10^{10}$ & $2.3\cdot 10^{11}$ & 4.2m & 42m & 10 & 27s & 13m & 28.9\\
        \bottomrule
      \end{tabular}
    \end{center}
    \end{adjustwidth}
\end{minipage}
\end{table*}

\begin{table*}[hbt!]
  \caption{Comparing with other Decomposition Methods}
  \label{table:other}
  \begin{adjustwidth}{-2cm}{-1cm}
    \begin{center}
      \scriptsize
    \begin{tabular}{llllll}
      \toprule
      & Algorithm \ref{alg:nondeg-f5tree} & {Singular}: Elimination Method & {Singular}: Algorithm in \cite{EisenbudHunekeVasconcelos_1992} & {Maple}: Regular Chains & {Macaulay2}\\
      \midrule
Cyclic 8 & 40m & segfault & >35h & >35h & >35h\\
Pseudo(2, 10) & 0.3s & 40s & >1h & >1h & >1h\\
Pseudo(2, 12) & 5.2s & >1h & >1h & >1h & >1h\\
Pseudo(2, 6) & 0.008s & <1s & <1s & 0.29s & 0.07s\\
Pseudo(2, 8) & 0.03s & <1s & 23m & 5.82s & 13.78s\\
Sing(2, 10) & 2.9s & >1h & >1h & >1h & >1h\\
Sing(2, 4) & 0.02s & 1s & >1h & 91.32s & 0.42s\\
Sing(2, 5) & 0.07s & 4s & >1h & >1h & 1.94s\\
Sing(2, 6) & 0.15s & 56s & >1h & >1h & 16.64s\\
Sing(2, 7) & 0.35s & 8m & >1h & >1h & 289s\\
Sing(2, 8) & 0.68s & 23m & >1h & >1h & >1h\\
Sing(2, 9) & 1.4s & >1h & >1h & >1h & >1h\\
Sos(2,4,2) & 0.03s & <1s & <1s & 19.4s & 0.16s\\
Sos(2,4,3) & 0.03s & 1s & 3m & 14m & 0.63s\\
Sos(2,5,2) & 0.02s & <1s & >1h & >1h & 0.37s\\
Sos(2,5,3) & 0.34s & >1h & >1h & >1h & 9.35s\\
Sos(2,5,4) & 7.3s & >1h & >1h & >1h & 183s\\
Sos(2,6,2) & 0.17s & <1s & >1h & >1h & 0.7s\\
Sos(2,6,3) & 1.4s & >1h & >1h & >1h & 107s\\
Sos(2,6,4) & 169s & >140m & >140m & >140m & >140m\\
Sos(2,6,5) & 43s & >1h & >1h & >1h & >1h\\
Sos(2,7,2) & 2.91s & <1s & >1h & 2.94s & 0.18s\\
Sos(2,7,3) & 41s & >1h & >1h & >1h & >1h\\
Sos(2,7,4) & 32m & >26h & segfault & >26h & >26h\\
Sos(2,7,5) & 20h & segfault & segfault & >200h & >200h\\
Sos(2,7,6) & 73h & segfault & segfault & >334h & >500h\\
Steiner & 42m & >50h & segfault & >50h & >50h\\
      \bottomrule
\end{tabular}
\end{center}
\end{adjustwidth}
\end{table*}

\printbibliography

\end{document}